\documentclass[10pt]{article}
\usepackage{amssymb,amsmath,latexsym}

\newtheorem{theorem}{Theorem}[section]
\newtheorem{lemma}[theorem]{Lemma}

\newtheorem{definition}[theorem]{Definition}
\newtheorem{example}{Example}

\newenvironment{proof}{\noindent
  \textbf{Proof.}}{\hfill$\Box$\\}
\newenvironment{proofidea}{\noindent
  \textbf{Proof idea.}}{\hfill$\Box$\\}
\newenvironment{proofsketch}{\noindent
  \textbf{Proof sketch.}}{\hfill$\Box$\\}
\newcommand{\cut}[1]{}

\newcommand{\instr}[5]{\ensuremath{\hbox to 60 pt
    {${#1}$\hfil${#2}$\hfil$ \rightarrow
      $\hfil${#3}$\hfil${#4}$\hfil${#5}$}}}

\newcommand{\cclass}[1]{\ensuremath{\mathbf{#1}}}
\newcommand{\bigo}[1]{\ensuremath{\mathcal{O}(#1)}}

\newcommand{\ecl}[1]{\ensuremath{\mathsf{ecl}(#1)}}

\newcommand{\D}{\ensuremath{\Delta}}
\newcommand{\G}{\ensuremath{\Gamma}}

\newcommand{\vp}{\ensuremath{\varphi}}

\newcommand{\fm}[1]{\emph{#1}}
\newcommand{\de}[1]{\emph{#1}}

\newcommand{\rel}[1]{\ensuremath{\mathcal{#1}}}

\newcommand{\power}[1]{\ensuremath{\mathcal{P}(#1)}}
\newcommand{\powerne}[1]{\ensuremath{\mathcal{P}^{\tiny +}(#1)}}
\newcommand{\union}{\, \cup \,}

\newcommand{\bigunion}{\bigcup \,}
\newcommand{\biginter}{\bigcap \,}

\newcommand{\inter}{\, \cap \,}

\newcommand{\crh}[2]{\ensuremath{\{\, #1 \mid \, #2\, \}}}

\newcommand{\set}[1]{\ensuremath{ \{ #1 \} }}

\newcommand{\lang}{\ensuremath{\mathcal{L}}}

\newcommand{\LTL}{\ensuremath{\textbf{LTL}}}

\newcommand{\con}{\wedge}

\newcommand{\imp}{\rightarrow}

\newcommand{\equivalence}{\leftrightarrow}

\newcommand{\ap}{\textbf{\texttt{AP}}}

\newcommand{\kframe}[1]{\ensuremath{\mathfrak{#1}}}

\newcommand{\mmodel}[1]{\ensuremath{\mathcal{#1}}}
\newcommand{\hintikka}[1]{\ensuremath{\mathcal{#1}}}

\newcommand{\sat}[3]{\ensuremath{\mmodel{#1}, #2 \Vdash #3}}

\newcommand{\cl}[1]{\ensuremath{\mathsf{cl}(#1)}}

\newcommand{\Rule}[1]{\textbf{(#1)}}

\newcommand{\ATL}{\textbf{ATL}}

\newcommand{\agents}{\ensuremath{\Sigma}}

\newcommand{\st}[1]{\ensuremath{\mathbf{states}(#1)}}

\newcommand{\brancharrow}{\ensuremath{\Longrightarrow}}

\newcommand{\tableau}[1]{\ensuremath{\mathcal{#1}}}

\newcommand{\knows}[1]{\ensuremath{\mathbf{K}}_{#1}}

\newcommand{\common}[1]{\ensuremath{\mathbf{C}}_{#1}}
\newcommand{\distrib}[1]{\ensuremath{\mathbf{D}}_{#1}}
\newcommand{\maelcd}{\textbf{MAEL}(CD)}
\newcommand{\cmaelcd}{\textbf{CMAEL(CD)}}

\newcommand{\psmodel}[1]{\ensuremath{\mathcal{M}^*}}
\newcommand{\pseudomodel}[1]{\ensuremath{\mathcal{M}^{**}}}

\begin{document}

\title{Tableau-based procedure for deciding satisfiability in the full
  coalitional multiagent epistemic logic}

\author{Valentin Goranko\footnote{School of Mathematics, University of
    the Witwatersrand, South Africa, \texttt{goranko@maths.wits.ac.za}}
  \and Dmitry Shkatov\footnote{School of
  Computer Science, University of the Witwatersrand, South Africa,
  \texttt{dmitry@cs.wits.ac.za}}}

\date{}

\maketitle

\begin{abstract}
  We study the multiagent epistemic logic \cmaelcd\ with operators for
  common and distributed knowledge for all coalitions of agents. We
  introduce Hintikka structures for this logic and prove that
  satisfiability in such structures is equivalent to satisfiability in
  standard models. Using this result, we design an incremental tableau
  based decision procedure for testing satisfiability in \cmaelcd.
\end{abstract}

\section{Introduction}
\label{sec:intro}

Over the last two decades, multiagent epistemic logics have been found
to be a useful tool for a variety of applications in computer science
and AI (\cite{Fagin95knowledge}, \cite{HoekMeyer95}), the main among
them being design, specification, and verification of distributed
protocols (\cite{HM90}, \cite{Halpern87}, \cite{FHV92}). In this
paper, we consider the \emph{full coalitional multiagent epistemic
  logic}, involving modal operators for \emph{individual} knowledge
for each agent\footnote{The notion of agents used in this paper is
  abstract; for example, agents can be thought of as components of a
  distributed system.}, as well as operators for \emph{common} and
\emph{distributed} knowledge among any (non-empty) \emph{coalition} of
agents; we call that logic \cmaelcd\footnote{Abbreviated from
  \textbf{C}oalitional \textbf{M}ulti\textbf{A}gent \textbf{E}pistemic
  \textbf{L}ogic with \textbf{C}ommon and \textbf{D}istributed
  knowledge.}.  Most of the multiagent epistemic logics studied so far
only cover fragments of \cmaelcd; e.g., the logic considered in
\cite{FHV92} contains, besides the individual knowledge modalities,
the operator of distributed knowledge only for the whole set of agents
in the language, while \cite{vdHM97} extends that system with common
knowledge operator for the whole set of agents. As far as we know, no
provably complete deductive system or a decision procedure has been
developed so far for \cmaelcd, although \cite{Fagin95knowledge}
propose (without proof) an axiomatic system which is presumed to be
complete for this logic.

One of the major issues in applying multiagent epistemic logics to
design of distributed systems is the development of algorithms for
constructive checking of formulae of those logics for satisfiability,
i.e, checking if a formula is satisfiable and, if so, constructing a
model for it. The main purpose of this paper is to develop a
tableau-based algorithm for the constructive satisfiability problem
for \cmaelcd.  In the recent precursor (\cite{GorSh08a}) to the
present paper, we have developed such an algorithm for the multiagent
epistemic logic with operators of individual knowledge, as well as
common and distributed knowledge for the set of \emph{all agents}. In
the present paper, we extend the results of \cite{GorSh08a} to
\maelcd.  The main challenge in such an extension lies in handling the
operators of distributed knowledge parameterized by coalitions of
agents.  Thus, even though the style of the tableaux presented here is
similar to the ones from \cite{GorSh08a}, the proof of correctness of
the procedure required more involved model-theoretic techniques
building on those used in~\cite{FHV92}.  Consequently, the present
paper substantially focuses on overcoming challenges raised by the
presence in the language of coalitional distributed knowledge
modalities.

The satisfiability-checking algorithms of both \cite{GorSh08a} and the
present paper are based on the incremental tableaux in the style first
proposed in~\cite{Wolper85} and adapted recently to logics of
strategic ability in multiagent systems in~\cite{GorSh08}.  Besides
our conviction that this approach to building decision procedures for
logics of multiagent systems is practically most optimal, the
uniformity of method and style of these tableaux is deliberate, as it
reflects our intention to eventually integrate them into a
tableau-based decision procedure for comprehensive logical systems for
reasoning about knowledge, time, and strategic abilities of agents and
coalitions in multiagent systems.

The structure of the paper is as follows: Section
\ref{sec:SyntaxSemantics} presents the syntax and semantics of
\cmaelcd; in Sections \ref{sec:HintikkaStructures} and
\ref{sec:EquivalenceCMAEMs}, we introduce Hintikka structures for this
logic and prove that satisfiability in Hintikka structures is
equivalent to satisfiability in models. Then, in Sections
\ref{sec:Tableau} and \ref{sec:completeness}, we develop the tableau
procedure for testing satisfiability of \cmaelcd-formulae and sketch
proofs of its soundless, completeness, and termination, and briefly
estimate its complexity. We illustrate our tableau procedure with two
examples in the Appendix.

\section{Syntax and semantics of the logic \cmaelcd}
\label{sec:SyntaxSemantics}


The language $\lang$ of \cmaelcd\ contains a (finite or countable) set
\ap\ of \de{atomic propositions}, whose arbitrary members we typically
denote by $p, q, r, \ldots$; a finite, non-empty set $\agents$ of
names of \de{agents}, whose arbitrary members we typically denote by
$a, b \ldots$ and whose subsets, called \emph{coalitions}, we
typically denote by $A, B, \ldots$ (possibly with decorations); a
sufficient repertoire of the Boolean connectives; and, for every
non-empty coalition $A$, the modal operators $\distrib{A}$ (``\emph{it
  is distributed knowledge among $A$ that \ldots}'') and $\common{A}$
(``\emph{it is common knowledge among $A$ that \ldots}''). The
formulae of $\lang$ are thus defined as follows:
\[\vp := p \mid \neg \vp \mid (\vp_1 \con \vp_2) \mid \distrib{A}
\vp \mid \common{A} \vp,\] where $p$ ranges over $\ap$ and $A$ ranges
over non-empty subsets of $\agents$; the set of all such subsets will
henceforth be denoted by $\powerne{\agents}$.  The other Boolean
connectives can be defined as usual. We denote formulae of $\lang$ by
$\vp, \psi, \chi, \ldots$ (possibly with decorations) and omit
parentheses in formulae whenever it does not result in ambiguity.  We
write $\vp \in \lang$ to mean that $\vp$ is a formula of $\lang$.

The distributed knowledge operator $\distrib{A} \vp$ intuitively means
that a ``superagent'', somebody who knows everything that any of the
agents in $A$ knows, can obtain $\vp$ as a logical consequence of his
knowledge. For example, if agent $a$ knows $\psi$ and agent $b$ knows
$\psi \imp \chi$, then $\distrib{\{a,b\}} \chi$ is true even though
neither $a$ nor $b$ knows $\chi$. The operators of individual
knowledge $\knows{a} \vp$ (``\emph{agent $a$ knows that $\vp$}''), for
$a \in \agents$, can then be defined as $\distrib{\set{a}} \vp$,
henceforth written $\distrib{a} \vp$.

The common knowledge operator $\common{A} \vp$ means that $\vp$ is
``public knowledge'' among $A$, i.e., that every agent in $A$ knows
$\vp$, and knows that every agent in $A$ knows $\vp$, etc. For
example, it is common knowledge among drivers that green light means
``go'' and red light means ``stop''.
Formulae of the form $\neg \common{A} \vp$ are referred to as
\de{(epistemic) eventualities}, for the reasons given later on.

Formulae of $\lang$ are interpreted over \emph{coalitional multiagent
  epistemic models}.  In this paper, we also need the auxiliary
notions of \emph{coalitional multiagent epistemic structures and
  frames}, which we now define.

\begin{definition}
  \label{def:cmaes}
  A \de{coalitional multiagent epistemic structure} (CMAES, for
  short) is a tuple $\kframe{G} = (\agents, S, \set{\rel{R}^D_A}_{A
    \in \powerne{\agents}}, \set{\rel{R}^C_A}_{A \in
    \powerne{\agents}})$, where
  \begin{enumerate}
  \item $\agents$ is a finite, non-empty set of agents;
  \item $S \ne \emptyset$ is a set of \de{states};
  \item for every $A \in
    \powerne{\agents}$, $\rel{R}^D_A$ is a binary relation on $S$;
  \item for every $A \in \powerne{\agents}$, $\rel{R}^C_A$ is the
    reflexive, transitive closure of $\bigunion_{A'\subseteq A}
    R^D_{A'}$.
  \end{enumerate}
\end{definition}

\begin{definition}
  \label{def:cmaef}
  A \de{coalitional multiagent epistemic frame} (CMAEF) is
  a CMAES $\kframe{F} = (\agents, S, \set{\rel{R}^D_A}_{A \in
    \powerne{\agents}}, \set{\rel{R}^C_A}_{A \in \powerne{\agents}})$,
  where each $\rel{R}^D_A$ is an equivalence relation satisfying the
  following condition:
  \begin{center}
    \(
    \begin{array}{lcl}
      (\dag) & & \rel{R}^D_A = \biginter_{a \in A} \rel{R}^D_{a}
    \end{array}
    \)
  \end{center}
  If condition (\dag) above is replaced by the following, weaker, one:
  \begin{center}
    \(
    \begin{array}{lcl}
      (\dag \dag) & & \rel{R}^D_A \subseteq \rel{R}^D_B \mbox{ whenever } B \subseteq A,
    \end{array}
    \)
  \end{center}
  then \kframe{F} is a \de{coalitional multiagent epistemic
    pseudo-frame} (pseudo-CMAEF).
\end{definition}

Note that in every (pseudo-)CMAEF $\rel{R}^D_A \subseteq \biginter_{a
  \in A} \rel{R}^D_{a}$, and hence \linebreak $\bigunion_{A'\subseteq
  A} R^D_{A'} = \bigunion_{a \in A} R^D_{a}$. Thus, condition (4) of
Definition~\ref{def:cmaes} is equivalent to requiring that, in
(pseudo-)CMAEFs, $\rel{R}^C_A$ is the transitive closure of
$\bigunion_{a \in A} R^D_{a}$, for every $A \in \powerne{\agents}$.
Moreover, each $\rel{R}^C_A$ in a (pseudo-)CMAEF is an equivalence
relation.

\begin{definition}
  \label{def:cmaem}
  A \de{coalitional multiagent epistemic model} (CMAEM) is a tuple
  $\mmodel{M} = (\kframe{F}, \ap, L)$, where $\kframe{F}$ is a CMAEF,
  $\ap$ is a set of atomic propositions, and $L: S \mapsto
  \power{\ap}$ is a \de{labeling function}, assigning to every state
  $s$ the set $L(s)$ of atomic propositions true at $s$.  If
  $\kframe{F}$ is a pseudo-CMAEF, then $\mmodel{M} = (\kframe{F}, \ap,
  L)$ is a \de{multiagent coalitional pseudo-model} (pseudo-CMAEM).
\end{definition}

The satisfaction relation between (pseudo-)CMAEMs, states, and
formulae is defined in the standard way.  In particular,
\begin{itemize}
\item \sat{M}{s}{\distrib{A} \vp} iff $(s, t) \in \rel{R}^D_A$ implies
  \sat{M}{t}{\vp};
\item \sat{M}{s}{\common{A} \vp} iff $(s, t) \in \rel{R}^C_A$ implies
  \sat{M}{t}{\vp}.
\end{itemize}

\begin{definition}
  Given a (pseudo-)CMAEM $\mmodel{M}$ and $\vp \in \lang$, we say that
  $\vp$ is \de{satisfiable} in $\mmodel{M}$ if \sat{M}{s}{\vp} holds
  for some $s \in \mmodel{M}$ and say that $\vp$ is \de{valid} in
  $\mmodel{M}$ if \sat{M}{s}{\vp} holds for every $s \in \mmodel{M}$.
  Satisfiability and validity in a class of (pseudo-)models are
  defined accordingly.
\end{definition}

The truth condition for the operator $\common{A}$ can be re-stated in
terms of reachability.  Let $\kframe{F}$ be a (pseudo-)CMAEF with
state space $S$ and let $s, t \in S$.  We say that $t$ is
\de{$A$-reachable from $s$} if either $s = t$ or, for some $n \geq 1$,
there exists a sequence $s = s_0 , s_1, \ldots, s_{n-1}, s_n = t$ of
elements of $S$ such that, for every $0 \leq i < n$, there exists $a
\in A$ such that $(s_i, s_{i+1}) \in R^D_a$.  It is then easy to see
that the following truth condition for $\common{A}$ is equivalent to
the one given above:
\begin{itemize}
\item \sat{M}{s}{\common{A} \vp} iff \sat{M}{t}{\vp} whenever $t$ is
  $A$-reachable from $s$.
\end{itemize}

Note also, that if $\agents = \set{a}$, then $\distrib{a} \vp
\equivalence \common{a} \vp$ is valid for all $\vp \in \lang$.  Thus,
the single-agent case is trivialized and, therefore, we assume
throughout the remainder of the paper that $\agents$ contains at least
2 agents.
\section{Hintikka structures}
\label{sec:HintikkaStructures}

Despite our ultimate interest in satisfiability of finite sets of
formulae in CMAEMs, the tableaux we present check for the existence of
a more general kind of semantic structure for $\Theta$ than a model,
namely a \emph{Hintikka structure}.  In this section, we show that
Hintikka structures satisfy the same sets of formulae as
pseudo-CMAEMs; in the next section, we show that CMAEMs satisfy the
same sets of formulae as pseudo-CMAEMs.  Consequently, testing for
satisfiability in a Hintikka structure can replace testing for
satisfiability in a CMAEM.  In the following discussion, for brevity,
we only consider single formulae; the extension to finite sets of
formulae is straightforward.

The most fundamental difference between (pseudo-)models and Hintikka
structures is that while the former specify the truth value of every
formula of \lang\ at each state, the latter only do so for the
formulae relevant to the evaluation of a fixed formula $\theta$.
Another important difference is that the accessibility relations in
(pseudo-) models must satisfy the explicitly stated conditions of
Definition~\ref{def:cmaef}, while in Hintikka structures conditions
are only imposed on the labels of the states in such a way that every
Hintikka structure generates, through the construction of
Lemma~\ref{lm:hintikka_to_psmodels} below, a pseudo-CMAEM so that the
``truth'' of formulae in the labels is preserved in the resultant
pseudo-model. To define Hintikka structures, we need the following
auxiliary notion.
\begin{definition}
  \label{def:fully_expanded}
  A set $\D \subseteq \lang$ is \de{fully expanded} if it satisfies
  the following conditions:
  \begin{itemize}
  \item if $\neg \neg \vp \in \D$, then $\vp \in \D$;
  \item if $\vp \con \psi \in \D$, then $\vp \in \D$ and $\psi \in
    \D$;
  \item if $\neg (\vp \con \psi) \in \D$, then $\neg \vp \in \D$ or
    $\neg \vp \in \D$;
  \item if $\distrib{A} \vp \in \D$, then $\distrib{A'} \vp \in \D$
    for every $A'$ such that $A \subseteq A' \subseteq \agents$;
  \item if $\distrib{A} \vp \in \D$, then $\vp \in \D$;
  \item if $\common{A} \vp \in \D$, then $\distrib{a} (\vp \con
    \common{A} \vp) \in \D$ for every $a \in A$;
  \item if $\neg \common{A} \vp \in \D$, then $\neg \distrib{a} (\vp
    \con \common{A} \vp) \in \D$ for some $a \in A$;
  \item if $\neg \distrib{A} \neg \distrib{B} \vp \in \D$, then
    $\distrib{(A \inter B)} \vp \in \D$.
  \end{itemize}
\end{definition}

\begin{definition}
  \label{def:hs}
  A \de{coalitional multi-agent epistemic Hintikka structure} (CMAEHS
  for short) is a tuple $(\agents, S, \set{\rel{R}^D_A}_{A \in
    \powerne{\agents}}, \set{\rel{R}^C_A}_{A \in \powerne{\agents}},
  H)$ such that \newline $(\agents, S, \set{\rel{R}^D_A}_{A \in
    \powerne{\agents}}, \set{\rel{R}^C_A}_{A \in \powerne{\agents}})$
  is a CMAES, and $H$ is a labeling of the elements of $S$ with sets
  of formulae of \lang\ satisfying the following constraints:
    \begin{description}
    \item[H1] if $\neg \vp \in H(s)$, then $\vp \notin H(s)$, for
      every $s \in S$;
    \item[H2] $H(s)$ is fully expanded, for every $s \in S$;
    \item[H3] if $\neg \distrib{A} \vp \in H(s)$, then there exists $t
      \in S$ such that $(s, t) \in \rel{R}^D_A$ and $\neg \vp \in
      H(t)$;
    \item[H4] if $(s, t) \in \rel{R}^D_A$, then $\distrib{A'} \vp \in
      H(s)$ iff $\distrib{A'} \vp \in H(t)$, for every $A' \subseteq
      A$;
    \item[H5] if $\neg \common{A} \vp \in H(s)$, then there exists $t
      \in S$ such that $(s, t) \in \rel{R}^C_A$ and $\neg \vp \in
      H(t)$.
    \end{description}
\end{definition}

\begin{definition}
  Let $\theta \in \lang$, $\Theta \subseteq \lang$, and $\hintikka{H}$
  be a CMAEHS with state space $S$.  We say that \hintikka{H} is a
  \de{CMAEHS for $\theta$}, or that $\theta$ is \de{satisfiable} in
  \hintikka{H}, if $\theta \in H(s)$ for some $s \in S$; we say that
  $\Theta$ is satisfiable in \hintikka{H} if $\Theta \subseteq H(s)$.
\end{definition}

We now prove that $\theta \in \lang$ is satisfiable in a pseudo-CMAEM
iff there exists a CMAEHS for $\theta$.  Given a pseudo-CMAEM
$\mmodel{M} = (\agents, S, \set{\rel{R}^D_A}_{A \in
  \powerne{\agents}},$ \linebreak $\set{\rel{R}^C_A}_{A \in
  \powerne{\agents}}, \ap, L)$, we define the \emph{extended labeling
function} $L^+: S \mapsto \power{\lang}$ on \mmodel{M} as follows:
$L^+(s) = \crh{\vp}{\sat{M}{s}{\vp}}$.  Then, it is routine to check
the following.

\begin{lemma}
  \label{lm:psmodels_to_hintikka}
  Let $\mmodel{M} = (\agents, S, \set{\rel{R}^D_A}_{A \in
    \powerne{\agents}}, \set{\rel{R}^C_A}_{A \in \powerne{\agents}},
  \ap, L)$ be a pseudo-CMAEM satisfying $\theta$ and let $L^+$ be the
  extended labeling on \mmodel{M}.  Then, $(\agents, S,$ \linebreak
  $\set{\rel{R}^D_A}_{A \in \powerne{\agents}}, \set{\rel{R}^C_A}_{A \in \powerne{\agents}}, L^+)$ is a CMAEHS for $\theta$.
\end{lemma}

Now, we argue in the opposite direction.

\begin{lemma}
  \label{lm:hintikka_to_psmodels}
  Let $\theta \in \lang$ be such that there exists a CMAEHS for
  $\theta$.  Then, $\theta$ is satisfiable in a pseudo-CMAEM.
\end{lemma}

\begin{proof}
  Let $\hintikka{H} = (\agents, S, \set{\rel{R}^D_A}_{A \in
    \powerne{\agents}}, \set{\rel{R}^C_A}_{A \in \powerne{\agents}},
  H)$ be an CMAEHS for $\theta$.  We construct a pseudo-CMAEM
  \mmodel{M'} satisfying $\theta$ out of \hintikka{H} as follows.

  First, for every $A \in \powerne{\agents}$, let $\rel{R}'^D_A$ be
  the reflexive, symmetric, and transitive closure of $\bigunion_{A
    \subseteq B} \rel{R}^D_{B}$ and let $\rel{R}'^C_A$ be the
  transitive closure of $\bigunion_{a \in A} \rel{R}'^D_a$. Notice
  that $\rel{R}^D_A \subseteq \rel{R}'^D_A$ and $\rel{R}^C_A \subseteq
  \rel{R}'^C_A$ for every $A \in \powerne{\agents}$. Second, let $L(s)
  = H(s) \inter \ap$, for every $s \in S$.  It is then easy to check
  that $\mmodel{M}' = (\agents, S, \set{\rel{R}'^D_A}_{A \in
    \powerne{\agents}}, \set{\rel{R}'^C_A}_{A \in \powerne{\agents}},
  \ap, L)$ is a pseudo-CMAEM.

  To complete the proof of the lemma, we show, by induction on the
  structure of $\chi \in \lang$ that, for every $s \in S$ and every
  $\chi \in \lang$, the following hold:

  \textbf{(i)} $\chi \in H(s) \text{ implies } \sat{M'}{s}{\chi}$;

  \textbf{(ii)} $\neg \chi \in H(s) \text{ implies } \sat{M'}{s}{\neg
    \chi}$.

  Let $\chi$ be some $p \in \ap$.  Then, $p \in H(s)$ implies $p \in
  L(s)$ and, thus, \sat{M'}{s}{p}; if, on the other hand, $\neg p \in
  H(s)$, then due to (H1), $p \notin H(s)$ and thus $p \notin L(s)$;
  hence, \sat{M'}{s}{\neg p}.

  Assume that the claim holds for all subformulae of $\chi$; then, we have
  to prove that it holds for $\chi$, as well.

  Suppose that $\chi$ is $\neg \vp$.  If $\neg \vp \in H(s)$, then the
  inductive hypothesis immediately gives us $\sat{M'}{s}{\neg \vp}$;
  if, on the other hand, $\neg \neg \vp \in H(s)$, then by virtue of
  (H2), $\vp \in H(s)$ and hence, by inductive hypothesis,
  $\sat{M'}{s}{\vp}$ and thus $\sat{M'}{s}{\neg \neg \vp}$.

  The case of $\chi = \vp \con \psi$ is straightforward, using (H2).

  Suppose that $\chi$ is $\distrib{A} \vp$.  Assume, first, that
  $\distrib{A} \vp \in H(s)$.  In view of the inductive hypothesis, it
  suffices to show that $(s, t) \in \rel{R}'^D_A$ implies $t \in
  H(t)$.  So, assume that $(s, t) \in \rel{R}'^D_A$.  There are two
  cases to consider.  If $s = t$, then the conclusion immediately
  follows from (H2).  If, on the other hand, $s \ne t$, then there
  exists an undirected path from $s$ to $t$ along the relations
  $\rel{R}^D_{A'}$, where each $A'$ is a superset of $A$.  Then, in
  view of (H4), $\distrib{A} \vp \in H(t)$; hence, by (H2), $\vp \in
  H(t)$, as desired.

  Assume, next, that $\neg \distrib{A} \vp \in H(s)$.  In view of the
  inductive hypothesis, it suffices to show that there exist $t \in
  S$ such that $(s, t) \in \rel{R}'^D_A$ and $\neg \vp \in H(t)$.  By
  (H3), there exists $t \in S$ such that $(s, t) \in \rel{R}^D_A$ and
  $\neg \vp \in H(t)$.  As $\rel{R}^D_A \subseteq \rel{R}'^D_A$, the
  desired conclusion follows.

  Suppose now that $\chi$ is $\common{A} \vp$.  Assume that
  $\common{A} \vp \in H(s)$.  In view of the inductive hypothesis, it
  suffices to show that if $t$ is $A$-reachable from $s$ in
  $\mmodel{M}'$, then $\vp \in H(t)$. So, assume that either $s = t$
  or, for some $n \geq 1$, there exists a sequence of states $s = s_0,
  s_1, \ldots, s_{n-1}, s_n = t$ such that, for every $0 \leq i < n$,
  there exists $a_i \in \agents$ such that $(s_i, s_{i+1}) \in
  \rel{R}'^D_{a_i}$. In the former case, the desired conclusion
  follows from (H2). In the latter case, we can show by induction on
  $0 \leq i < n$ that $\distrib{a_i} (\vp \con \common{A} \vp) \in
  H(s_i)$.  Then, $\distrib{a_{n-1}} (\vp \con \common{A} \vp) \in
  H(s_{n-1})$, and thus, in view of (H3) and (H2), $\vp \in H(t)$.

  Assume, on the other hand, that $\neg \common{A} \vp \notin H(s)$.
  Then, the desired conclusion follows from (H6), the fact that
  $\rel{R}^C_A \subseteq \rel{R}'^C_A$, and the inductive hypothesis.
  
\end{proof}

\begin{theorem}
  \label{thr:psmod_equals_hintikka}
  Let $\theta \in \lang$.  Then, $\theta$ is satisfiable in a
  pseudo-CMAEM iff there exists a CMAEHS for $\theta$.
\end{theorem}

\begin{proof}
  Immediately follows from Lemmas~\ref{lm:psmodels_to_hintikka} and
  \ref{lm:hintikka_to_psmodels}. 
\end{proof}

\section{Equivalence of CMAEMs and pseudo-CMAEMs}
\label{sec:EquivalenceCMAEMs}

In the present section, we prove that pseudo-CMAEMs and CMAEMs satisfy
the same sets of formulae.  The right-to-left direction is immediate,
as every CMAEM is a pseudo-CMAEM.  For the left-to-right direction, we
use a modification of the construction from~\cite[appendix A1]{FHV92}
to show that if $\theta \in \lang$ is satisfiable in a pseudo-CMAEM,
then it is satisfiable in a ``tree-like'' pseudo-CMAEM that actually
turns out to be a bona-fide CMAEM.

\begin{definition}
  Let $\mmodel{M} = (\agents, S, \set{\rel{R}^D_A}_{A \in
    \powerne{\agents}}, \set{\rel{R}^C_A}_{A \in \powerne{\agents}},
  \ap, L)$ be a \linebreak (pseudo-) CMAEM and let $s, t \in S$.  A
  \de{maximal path from $s$ to $t$} in \mmodel{M} is a sequence $s =
  s_0, A_0, s_1, \ldots, s_{n-1}, A_{n-1}, s_n = t$ such that, for
  every $0 \leq i < n$, $(s_i, s_{i+1}) \in \rel{R}^D_{A_i}$, but
  $(s_i, s_{i+1}) \in \rel{R}^D_{B}$ does not hold for any $B$ with
  $A_i \subset B \subseteq \agents$.  A segment $\rho'$ of a maximal
  path $\rho$ starting and ending with a state is a \de{sub-path} of
  $\rho$.
\end{definition}

\begin{definition}
  \label{def:reduced_paths}
  Let $\rho = s_0, A_0 \ldots, A_{n-1}, s_n$ be a maximal path in
  \mmodel{M}.  The \de{reduction} of $\rho$ is obtained by, first,
  replacing in $\rho$ every longest sub-path $s_p, A_p , s_{p+1},
  \ldots, \linebreak A_{p+q-1}, s_{p+q}$ such that $s_p = s_{p+1} =
  \ldots = s_{p+q}$ with $s_p$ (i.e., removing loops) and, then, by
  replacing in the resultant path every longest sub-path $s_j, A_j,
  s_{j+1}, \linebreak \ldots, A_{j+m-1}, s_{j+m}$ such that $A_j =
  A_{j+1} = \ldots = A_{j+m-1}$ with $s_j, A_j, s_{j+m}$ (i.e.,
  collapsing multiple consecutive transitions along the same relation
  with a single transition).  A maximal path is \de{reduced} if it
  equals its own reduction.
\end{definition}

\begin{definition}
  \label{def:tree-like-models}
  A (pseudo-)CMAEM \mmodel{M} is \de{tree-like} if, for every $s, t
  \in \mmodel{M}$, there exists at most one reduced maximal path from
  $s$ to $t$.
\end{definition}

\begin{lemma}
  \label{lm:tree_like_model}
  If $\theta \in \lang$ is satisfiable in a pseudo-CMAEM, then it is
  satisfiable in a (tree-like) CMAEM.
\end{lemma}

\begin{proof}
  Suppose that $\theta$ is satisfied in a pseudo-CMAEM $\mmodel{M} =
  (\agents, S, \set{\rel{R}^D_A}_{A \in \powerne{\agents}}, \linebreak
  \set{\rel{R}^C_A}_{A \in \powerne{\agents}}, \ap, L)$ at state $s$.
  To build a tree-like CMAEM satisfying $\theta$, we use a slight
  modification of the standard technique of tree-unraveling.  The only
  difference between our construction and the standard tree-unraveling
  is that the state space of our tree model is made up of all
  \de{maximal} paths in \mmodel{M} rather than all paths whatsoever.

  Let $\mmodel{M}' = (\agents, S', \set{\rel{R}'^D_A}_{A \in
    \powerne{\agents}}, \set{\rel{R}'^C_A}_{A \in \powerne{\agents}},
  \ap, L')$ be the submodel of \mmodel{M} generated by $s$.  Then,
  $\sat{M'}{s}{\theta}$ since $\mmodel{M}$ and $\mmodel{M}'$ are
  locally bisimilar at $s$. Next, we unravel $\mmodel{M}'$ into a
  model $\mmodel{M}'' = (\agents, S'', \set{\rel{R}''^D_A}_{A \in
    \powerne{\agents}}, \set{\rel{R}''^C_A}_{A \in \powerne{\agents}},
  \linebreak \ap, L'')$ as follows.  First, call a maximal path $\rho$
  in $\mmodel{M}'$ an $s$-max-path if the first component of $\rho$ is
  $s$.  Denote \cut{the length of an $s$-\emph{max-path} $\rho$,
    defined as the number of states in $\rho$, by $len(\rho)$ and} the
  last element of $\rho$ by $l(\rho)$. Notice that $s$ by itself is an
  $s$-max-path.  Now, let $S^{''}$ be the set of all $s$-max-paths in
  $\mmodel{M}'$.  For every $A \in \powerne{\agents}$, let
  $\rel{R}^*{}^D_A$ be \crh{(\rho, \rho')} {\rho, \rho' \in S''
    \text{and } \rho' = \rho, A, l(\rho')} and let, furthermore,
  $\rel{R}''^D_A$ to be the reflexive, symmetric, and transitive
  closure of $\rel{R}^*{}^D_A$. Notice that $(\rho,\rho')\in
  \rel{R}''^D_A$ holds iff one of the paths $\rho$ and $\rho'$ extends
  the other by a sequence of $A$-steps. Therefore, two different
  states in $S''$ can only be connected by $\rel{R}''^D_A$ for at most
  one maximal coalition $A$.  Further, we stipulate the following
  \emph{downwards closure condition}: whenever $(\rho, \tau) \in
  \rel{R}''^D_A$ and $B \subseteq A$, then $(\rho, \tau) \in
  \rel{R}''^D_B$.  The relations $\rel{R}''^C_A$ are then defined as
  in any CMAEF.  To complete the definition of $\mmodel{M}''$, we put
  $L''(\rho) = L'(l(\rho))$, for every $\rho \in S''$.

  It is clear from the construction that $\mmodel{M}''$ is a
  pseudo-CMAEM.  We will now show that it actually is a (tree-like)
  CMAEM and that it satisfies $\theta$.  To prove the first part of
  the claim, we need some extra terminology.

  We call a maximal path $\rho_1, A_1, \rho_2, \ldots, A_{n-1},
  \rho_n$ in $\mmodel{M}''$ \de{primitive} if, for every $0 \leq i <
  n$, either $(\rho_i, \rho_{i+1}) \in \rel{R}^*{}^D_{A_i}$ or
  $(\rho_{i+1}, \rho_i) \in \rel{R}^*{}^D_{A_i}$.  A primitive path
  $\rho_1, A_1, \rho_2, \ldots, A_{n-1}, \rho_n$ is \de{non-redundant}
  if there is no $0 \leq i < n$ such that $\rho_i = \rho_{i+2}$ and
  $A_i = A_{i+1}$. Intuitively, in a non-redundant path we never go
  from a state $\rho$ (forward or backward) along a relation and then
  immediately back to $\rho$ along the same relation.  Since the
  relations $\rel{R}^*{}^D_A$ are edges of a tree, it immediately
  follows that:
  \begin{description}
  \item[(\ddag)] for every pair of states $\rho, \tau \in S''$, there
    exists at most one non-redundant primitive path from $\rho$ to
    $\tau$.
  \end{description}
  Lastly, we call a primitive path $\rho_1, A, \rho_2, \ldots, A,
  \rho_n$ an \de{$A$-primitive path}.

  We will now show that maximal reduced paths in $\mmodel{M}''$ stand
  in one-to-one correspondence with non-redundant primitive paths.  It
  will then follow from (\ddag) that maximal reduced paths between any
  two states of $\mmodel{M}''$ are unique, and thus $\mmodel{M}''$ is
  tree-like, as claimed.  Let $P = \rho_1, A_1, \ldots, A_{n-1},
  \rho_n$, where $\rho_1 = \rho$ and $\rho_n = \tau$, be a maximal
  reduced path from $\rho$ to $\tau$ in $\mmodel{M}''$. Since
  $(\rho_i, \rho_{i+1}) \in \rel{R}''^D_{A_i}$, there exists a
  non-redundant $A_i$-primitive path from $\rho_i$ to $\rho_{i+1}$,
  which in view of (\ddag) is unique. Let us obtain a path $P'$ from
  $\rho$ to $\tau$ by replacing in $\rho$ every link $(\rho_i, A_i,
  \rho_{i+1})$ by the corresponding non-redundant $A_i$-primitive path
  from $\rho_i$ to $\rho_{i+1}$.  Call $P'$ an \de{expansion} of $P$.
  In view of (\ddag), every path has a unique expansion.  Now, it is
  easy to see that $P$ is a reduction of $P'$.  Since the reduction of
  a given path is unique, too, it follows that there exists a
  one-to-one correspondence between maximal reduced paths and
  non-redundant primitive paths in $\mmodel{M}''$.

  Next, we prove that $\rel{R}''^D_A = \biginter_{a \in A}
  \rel{R}''^D_a$ for every $A \in \powerne{\agents}$, and therefore
  that $\mmodel{M}''$ is a CMAEM.  The left to right inclusion is
  immediate from the construction (namely, because of the downward
  saturation condition we imposed on epistemic relations).  For the
  opposite direction, assume that $(s, t) \in \rel{R}''^D_a$ holds for
  every $a \in A$.  Then, for every $a \in A$, there exists a path,
  and therefore a maximal reduced path, from $s$ to $t$ along
  relations $\rel{R}^*{}^D_{A'}$ such that $a \in A'$.  As
  $\mmodel{M}''$ is tree-like, there is only one maximal reduced path
  from $s$ to $t$.  Therefore, the relations $\rel{R}^D_{A'}$ linking
  $s$ to $t$ along this path are such that $A \subseteq A'$ for every
  $A'$.  Then, by the downwards closure condition, there is a path
  from $s$ to $t$ along the relation $\rel{R}''^D_A$ and, hence, $(s,
  t) \in \rel{R}''^D_A$, as desired.

  Finally, it remains to prove that $\mmodel{M}''$ satisfies $\theta$.  First,
  notice that $(\rho, \tau) \in \rel{R}''_A$ iff there exists an
  $A$-primitive path from $\rho$ to $\tau$; hence, as every
  $\rel{R}'_A$ is an equivalence relation, if $(\rho, \tau) \in
  \rel{R}''_A$, then $(l(\rho), l(\tau)) \in \rel{R}'_A$. It is now
  easy to check that the relation $Z = \crh{(\rho, l(\rho)}{\rho \in
    S''}$ is a bisimulation between $\mmodel{M}''$ and $\mmodel{M}'$.
  Since $(s, l(s)) \in Z$, it follows that \sat{M''}{s}{\theta},
  and we are done. 
\end{proof}

\begin{theorem}
  \label{thr:models_equal_hintikka}
  Let $\theta \in \lang$.  Then, $\theta$ is satisfiable in a CMAEM
  iff there exists a Hintikka structure for $\theta$.
\end{theorem}

\begin{proof}
  Immediate from Theorem~\ref{thr:psmod_equals_hintikka} and Lemma
  \ref{lm:tree_like_model}. 
\end{proof}

\section{Tableaux for \cmaelcd}
\label{sec:Tableau}

\subsection{Basic ideas and overview of the tableau procedure}

The tableau procedure for testing a formula $\theta \in \lang$ for
satisfiability is an attempt to construct a non-empty graph
$\tableau{T}^{\theta}$ (called \fm{tableau}) representing \emph{all
  possible} CMAEHSs for $\theta$. The philosophy underlying our
tableau algorithm is essentially the same as the one underpinning the
tableau procedure for \LTL\ from~\cite{Wolper85}, recently adapted to
Alternating-time logic \ATL\ in~\cite{GorSh08} and to multiagent
epistemic logic with operators of common and distributed knowledge for
the whole set of agents in~\cite{GorSh08a}. To make the present paper
self-contained, we first outline the basic idea behind the tableau
algorithm for \cmaelcd, following ~\cite{GorSh08} and \cite{GorSh08a}.
The details of the tableaux presented here, obviously, are specific to
\cmaelcd.

Usually, tableaux check for satisfiability by decomposing the input
formula into ``semantically simpler'' formulae. In the classical
propositional case, ``semantically simpler'' implies ``smaller'', thus
ensuring the termination of the procedure. Another feature of the
tableaux for classical propositional logic is that the decomposition
into simpler formulae results in a simple tree, representing an
exhaustive search for a model---or, to be more precise, a Hintikka set
(the classical analogue of Hintikka structures)---for the input
formula $\theta$.  If at least one leaf of the tree produces a
Hintikka set for $\theta$, the search has succeeded and $\theta$ is
pronounced satisfiable.

These two defining features of the classical tableau method do not
directly apply to logics containing fixed point operators, such as
$\common{A}$. First, decomposing of a formula $\common{A} \vp$
produces the formulae of the form $\distrib{a} (\vp \con \common{A}
\vp)$, which are not exactly ``semantically simpler''; rather, the
unfolding of the monotone operator whose fixed point is $\common{A}
\vp$ is effected. Hence, we cannot take termination for granted and
need to put a mechanism in place that would guarantee it---in our
tableaux, this mechanism consists in the use of prestates, whose role
is to ensure the finiteness, and hence termination, of the
construction.  Second, in the classical case, the only reason why the
tableau might fail to produce a Hintikka set for the input formula is
that every attempt to build such a set results in a collection of
formulae containing a \emph{patent inconsistency} (a pair of formulae
$\vp, \neg \vp$).  In the case of \cmaelcd, there are other such
reasons, as the tableau is meant to represent CMAEHSs, which are more
complicated structures than classical Hintikka sets.  One additional
possible reason for a failure of a node of the tableau to be
satisfiable has to do with eventualities: the presence of an
eventuality $\neg \common{A} \vp$ in the label of a state $s$ of a
CMAEHS requires that there is an $A$-path from $s$ to a state $t$
whose label contains $\neg \vp$. The analogue of this condition in the
tableau is called \emph{realization of eventualities}.  Thus, all
eventualities in a tableau should be realized in order for the tableau
to be ``good'', i.e. to eventually produce a Hintikka structure. The
third possible reason for existence of ``bad'' nodes in the tableau
has to do with successor nodes---it may so happen that some of the
successors of a node $s$ whose satisfiability is necessary for the
satisfaction of $s$ itself are unsatisfiable; a ``good'' tableau
should not contain such ``bad'' nodes.

The tableau procedure consists of three major phases:
\fm{construction}, \fm{prestate elimination}, and \fm{state
  elimination}.  During the construction phase, we produce a directed
graph $\tableau{P}^{\theta}$---called the \emph{pretableau} for
$\theta$---whose set of nodes properly contains the set of nodes of
the tableau $\tableau{T}^{\theta}$ we are building.  Nodes of
$\tableau{P}^{\theta}$ are sets of formulae, some of which, called
\fm{states}, are meant to represent states of a Hintikka structure,
while others, called \fm{prestates}, play an auxiliary, technical role
in the construction of $\tableau{P}^{\theta}$.  During the prestate
elimination phase, we create a smaller graph $\tableau{T}_0^{\theta}$
out of $\tableau{P}^{\theta}$, called the \fm{initial tableau for
  $\theta$}, by eliminating all the prestates of
$\tableau{P}^{\theta}$ and adjusting its edges---prestates have
already fulfilled their role and can be discharged.  Finally, during
the state elimination phase, we remove from $\tableau{T}_0^{\theta}$
all the states, if any, that cannot be satisfied in any CMAEHS, for
one of the three above-mentioned reasons.  The elimination procedure
results in a (possibly empty) subgraph $\tableau{T}^{\theta}$ of
$\tableau{T}_0^{\theta}$, called the \de{final tableau for
  $\theta$}. If some state $\Delta$ of $\tableau{T}^{\theta}$ contains
$\theta$, we declare $\theta$ satisfiable; otherwise, we declare it
unsatisfiable.

The reader is referred to the examples given in the Appendix, to trace all phases of the construction of the tableau.

\subsection{Construction phase}

At this phase, we build the pretableau $\tableau{P}^{\theta}$---a
directed graph whose nodes are sets of formulae, coming in two
varieties: \fm{states} and \fm{prestates}.  States are meant to
represent states of CMAEHSs the tableau attempts to construct, while
prestates are ``embryo states'', expanded into states in the course of
the construction.  Formally, states are fully expanded (recall
Definition~\ref{def:fully_expanded}), while prestates do not have to
be so. Moreover, $\tableau{P}^{\theta}$ contains two types of edge.
As already mentioned, a tableau attempt to produce a compact
representation of all possible CMAEHSs for the input formula; in this
attempt, it organizes an exhaustive search for such CMAEHSs. One type
of edge, depicted by unmarked double arrows $\brancharrow$, represents
the expansion of the tableau as a search tree.  That exhaustive search
considers all possible alternatives, which arise when expanding
prestates into (fully expanded) states by branching in the disjunctive
cases. Thus, when we draw a double arrow from a prestate \G\ to states
$\D$ and $\D'$ (depicted as $\G \brancharrow \D$ and $\G \brancharrow
\D'$, respectively), this intuitively means that, in any CMAEHS, a
state satisfying \G\ has to satisfy at least one of $\D$ and
$\D'$. Our first construction rule, \Rule{SR}, prescribes how to
create states from prestates.

Given a set $\G \subseteq \lang$, we say that $\D$ is a \de{minimal
fully expanded extension of \G} if $\D$ is fully expanded, $\G
\subseteq \D$, and there is no $\D'$ such that $\G \subseteq \D'
\subset \D$ and $\D'$ is fully expanded.

\smallskip

\textbf{Rule} \Rule{SR} Given a prestate $\G$ such that \Rule{SR} has
not been applied to it before, do the following:
\begin{enumerate}
\item Add 
all minimal fully expanded extensions  $\D$ of $\G$ as \de{states};
\item For each so obtained state $\D$, put $\G \brancharrow \D$;
\item If, however, the pretableau already contains a state $\D'$ that
  coincides with $\D$, do not create another copy of $\D'$, but only
  put $\G \brancharrow \D'$.
\end{enumerate}

We denote by $\st{\G}$ the (finite) set \crh{\D}{\G \brancharrow \D}.

The second type of an edge featuring in our tableaux represents
transition relations in the CMAEHS which it attempts to build.
Accordingly, this type of edge is represented by single arrows marked
with formulae whose presence in the source state requires the presence
in the tableau of a target state, reachable by a particular relation.
All such formulae have the form $\neg \distrib{A} \vp$ (as can be seen
from Definition~\ref{def:hs}).
Intuitively, if, say, $\neg \distrib{A} \vp \in \D$, then we need some
prestate $\G$ containing $\neg \vp$ to be accessible from $\D$ by
$\rel{R}^D_A$; the reason we mark this single arrow not just by
coalition $A$, but by formula $\neg \distrib{A} \vp$, is that it helps
us remember not just what relation connects states satisfying $\D$ and
$\G$, but why we had to create this particular $\G$.  This information
will be needed when we start eliminating prestates and then states.

The second construction rule, \Rule{DR}, prescribes how to create
prestates from states.  This rule does not apply to states containing
patent inconsistencies (such sets are called \emph{patently
  inconsistent}), as such states cannot be satisfied in any CMAEHS.

\smallskip \textbf{Rule} \Rule{DR}: Given a state $\D$ such that $\neg
\distrib{A} \vp \in \D$, state $\D$ is not patently inconsistent, and
\Rule{DR} has not been applied to it before, do the following:

\begin{enumerate}
\item Create a new prestate $\G = \set{\neg \vp} \union \bigunion_{A'
    \subseteq A} \crh{\distrib{A'} \psi}{\distrib{A'} \psi \in \D}
  \union \\ \bigunion_{A' \subseteq A} \crh{\neg \distrib{A'}
    \psi}{\neg \distrib{A'} \psi \in \D}$;
\item Connect $\D$ to $\G$ with $\stackrel{\neg \distrib{A}
    \vp}{\longrightarrow}$;
\item If, however, the tableau already contains a prestate $\G' = \G$,
  do not add to it another copy of $\G'$, but simply connect $\D$ to
  $\G'$ with $\stackrel{\neg \distrib{A} \vp}{\longrightarrow}$.
\end{enumerate}


When building a tableau for $\theta \in \lang$, the construction phase
begins with creating a single prestate \set{\theta}.  Afterwards, we
alternate between \Rule{SR} and \Rule{DR}: first, \Rule{SR} is applied
to the prestates created at the previous stage of the construction,
then \Rule{DR} is applied to the states created at the previous stage.
The construction phase comes to an end when every prestate required to
be added to the pretableau has already been added (as prescribed in
point 3 of \Rule{SR}), or when we end up with states to which
\Rule{DR} does not apply.

Since we identify states and prestates whenever possible, to prove
termination of the construction phase it suffices to show that there
are only finitely many possible states and prestates. For that, we use
the concept of an extended closure of a formula.
\begin{definition}
  \label{def:extended_closure}
  Let $\theta \in \lang$.  The \de{closure} of $\theta$, denoted
  $\cl{\theta}$, is the least set of formulae such that $\theta \in
  \cl{\theta}$, $\cl{\theta}$ is closed under subformulae, and the
  following conditions hold:
  \begin{itemize}
  \item if $\distrib{A} \vp \in \cl{\theta}$ and $A \subseteq A'
    \subseteq \agents$, then $\distrib{A'} \vp \in \cl{\theta}$;
  \item if $\common{A} \vp \in \cl{\theta}$, then $\distrib{a} (\vp
    \con \common{A} \vp) \in \cl{\theta}$ for every $a \in A$.
  \end{itemize}
  The \de{extended closure} of $\theta$, denoted $\ecl{\theta}$, is
  the least set such that, if $\vp \in \cl{\theta}$, then $\vp, \neg
  \vp \in \ecl{\theta}$.
\end{definition}

It is straightforward to check that $\ecl{\theta}$ is finite for every
$\theta$ and that all states and prestates of $\tableau{P}^{\theta}$
are subsets of $\ecl{\theta}$; hence, their number is indeed finite;
hence, the construction phase terminates.

\subsection{Prestate elimination phase}

At this phase, we remove from $\tableau{P}^{\theta}$ all the prestates
and unmarked arrows, by applying the following rule:

\smallskip

\Rule{PR} For every prestate $\G$ in $\tableau{P}^{\theta}$, do the
following:

\begin{enumerate}
\item Remove $\G$ from $\tableau{P}^{\theta}$;
\item If there is a state $\D$ in $\tableau{P}^{\theta}$ with $\D
  \stackrel{\chi}{\longrightarrow} \G$, then for every state $\D' \in
  \st{\G}$, put $\D \stackrel{\chi}{\longrightarrow} \D'$;
\end{enumerate}

The resultant graph is denoted $\tableau{T}_0^{\theta}$ and called the
\emph{initial tableau}.

\subsection{State elimination phase}

During this phase, we remove from $\tableau{T}_0^{\theta}$ states that
are not satisfiable in any CMAEHS.  Recall, that there are three
reasons why a state $\D$ of $\tableau{T}_0^{\theta}$ can turn out to
be unsatisfiable: $\D$ is patently inconsistent, \emph{or}
satisfiability of $\D$ requires satisfiability of some other
unsatisfiable ``successor'' states, \emph{or} $\D$ contains an
eventuality that is not realized in the tableau. Accordingly, we have
three elimination rules, \Rule{E1}--\Rule{E3}.

Formally, the state elimination phase is divided into stages; we start
at stage 0 with $\tableau{T}_0^{\theta}$; at stage $n+1$, we remove
from the tableau $\tableau{T}_n^{\theta}$ obtained at the previous
stage exactly one state, by applying one of the elimination rules,
thus obtaining the tableau $\tableau{T}_{n+1}^{\theta}$. We state the
rules below, where $S_m^{\theta}$ denotes the set of states of
$\tableau{T}_{m}^{\theta}$.

\smallskip

\Rule{E1} If $\set{\vp, \neg \vp} \subseteq \D \in S_n^{\theta}$, then
obtain $\tableau{T}_{n+1}^{\theta}$ by eliminating $\D$ from
$\tableau{T}_n^{\theta}$.

\smallskip

\Rule{E2} If $\D$ contains a formula $\chi = \neg \distrib{A} \vp$ and
all states reachable from $\D$ by single arrows marked with $\chi$
have been eliminated at previous stages, obtain
$\tableau{T}_{n+1}^{\theta}$ by eliminating $\D$ from
$\tableau{T}_n^{\theta}$.

\smallskip

For the third elimination rule, we need the concept of
\emph{eventuality realization}.  We say that the eventuality $\xi=
\neg \common{A} \vp$ is realized at $\D$ in $\tableau{T}^{\theta}_n$
if either $\neg \vp \in \D$ or there exists in
$\tableau{T}^{\theta}_n$ a finite path $\D_0, \D_1, \ldots, \D_m$ such
that $\D_0 = \D$, $\neg \vp \in \D_m$, and for every $0 \leq i < m$
there exist $\chi_i = \distrib{B} \psi_i$ such that $B \subseteq A$
and $\D_i \stackrel{\chi_i}{\longrightarrow} \D_{i+1}$. We check for
realization of $\xi$ by running the following marking procedure that
marks all states that realize an eventuality $\xi$ in
$\tableau{T}_n^{\theta}$.  Initially, we mark all $\D \in
S_n^{\theta}$ such that $\neg \vp \in \D$. Then, we repeatedly do the
following: if $\D \in S_n^{\theta}$ is unmarked and there exists at
least one $\D'$ such that $\D \stackrel{\distrib{B}
  \psi}{\longrightarrow} \D'$ for some $B\subseteq A$ and $\D'$ is
marked, then $\D$ gets marked. The procedure ends when no more states
get marked at a current round of marking. Note that marking is carried
out with respect to a fixed eventuality $\xi$ and is, therefore,
repeated as many times as the number of eventualities in (the states)
of a tableau. Now, we can state our last rule.

\medskip

\Rule{E3} If $\D \in S_n^{\theta}$ contains an eventuality $\neg
\common{A} \vp$ that is not realized at $\D$ in
$\tableau{T}_n^{\theta}$, then obtain $\tableau{T}_{n+1}^{\theta}$ by
removing $\D$ from $\tableau{T}_n^{\theta}$.

\medskip

We have so far described individual rules; to describe the state
elimination phase as a whole, we must specify the order of their
application. First, we apply \Rule{E1} to all the states of
$\tableau{T}_0^{\theta}$; once this is done, we do not need to apply
\Rule{E1} again.  The cases of \Rule{E2} and \Rule{E3} are more
involved.  After having applied \Rule{E3} we could have removed all
the states accessible from some $\D$ along the arrows marked with some
formula $\chi$; hence, we need to reapply \Rule{E2} to the resultant
tableau to remove such $\D$'s. Conversely, after having applied
\Rule{E2}, we could have thrown away some states that were needed for
realizing certain eventualities; hence, we need to reapply \Rule{E3}.
Therefore, we need to apply \Rule{E3} and \Rule{E2} in a dovetailed
sequence that cycles through all the eventualities. More precisely, we
arrange all eventualities occurring in the tableau obtained from
$\tableau{T}_0^{\theta}$ after having applied \Rule{E1} in a list:
$\xi_1, \ldots, \xi_m$.  Then, we proceed in cycles. Each cycle
consists of alternatingly applying \Rule{E3} to the pending
eventuality (starting with $\xi_1$), and then applying \Rule{E2} to
the resulting tableau, until all the eventualities have been dealt
with.  These cycles are repeated until no state is removed in a whole
cycle. Then, the state elimination phase is over.

The graph produced at the end of the state elimination phase is called
the \fm{final tableau for $\theta$}, denoted by $\tableau{T}^{\theta}$
and its set of states is denoted by $S^{\theta}$.

\begin{definition}
  The final tableau $\tableau{T}^{\theta}$ is \de{open} if $\theta \in
  \D$ for some $\D \in S^{\theta}$; otherwise, $\tableau{T}^{\theta}$
  is \de{closed}.
\end{definition}

The tableau procedure returns ``no'' if the final tableau is closed;
otherwise, it returns ``yes'' and, moreover, provides sufficient
information for producing a finite model satisfying $\theta$; that
construction is sketched in Section \ref{sec:completeness}.

\section{Soundness, completeness, and complexity}
\label{sec:completeness}

The \emph{soundness} of a tableau procedure amounts to claiming that
if the input formula $\theta$ is satisfiable, then the tableau for
$\theta$ is open.  To establish soundness of the overall procedure, we
use a series of lemmas showing that every rule by itself is sound; the
soundness of the overall procedure is then an easy consequence. The
proofs of the following two lemmas are straightforward.

\begin{lemma}
  \label{lm:expansion}
  Let $\G$ be a prestate of $\tableau{P}^{\theta}$ such that
  \sat{M}{s}{\G} for some CMAEM \mmodel{M} and $s \in
  \mmodel{M}$.  Then, \sat{M}{s}{\D} holds for at least one $\D \in
  \st{\G}$.
\end{lemma}

\begin{lemma}
  \label{lm:DR_sound}
  Let $\D \in S_0^{\theta}$ be such that \sat{M}{s}{\D} for some CMAEM
  \mmodel{M} and $s \in \mmodel{M}$, and let $\neg \distrib{A} \vp \in
  \D$.  Then, there exists $t \in \mmodel{M}$ such that $(s, t) \in
  \rel{R}^D_A$ and \sat{M}{t}{\set{\neg \vp} \union \bigunion_{A'
      \subseteq A} \crh{\distrib{A'} \psi}{\distrib{A'} \psi \in \D}
    \union \bigunion_{A' \subseteq A} \crh{\neg \distrib{A'}
      \psi}{\neg \distrib{A'} \psi \in \D}}.
\end{lemma}

\begin{lemma}
  \label{lm:E3_sound}
  Let $\D \in S_0^{\theta}$ be such that \sat{M}{s}{\D} for some CMAEM
  \mmodel{M} and $s \in \mmodel{M}$, and let $\neg \common{A} \vp \in
  \D$.  Then, $\neg \common{A} \vp$ is realized at $\D$ in
  $\tableau{T}_n^{\theta}$.
\end{lemma}

\begin{proofidea}
  As $\neg \common{A} \vp$ is true at $s$, there is a path in
  \mmodel{M} from $s$ leading to a state satisfying $\neg \vp$.  As
  the tableaux organize the exhaustive search, a chain of tableau
  states corresponding to those states in the model will be produced.
\end{proofidea}
\begin{theorem}
  \label{thr:soundness}
  If $\theta \in \lang$ is satisfiable in a CMAEM, then
  $\tableau{T}^{\theta}$ is open.
\end{theorem}
\begin{proofsketch}
  Using the preceding lemmas, show by induction on the number of
  stages in the state elimination process that no satisfiable state
  can be eliminated due to \Rule{E1}--\Rule{E3}.  The claim then
  follows from Lemma~\ref{lm:expansion}.
\end{proofsketch}

The \emph{completeness} of a tableau procedure means that if the
tableau for a formula $\theta$ is open, then $\theta$ is satisfiable
in a CMAEM. In view of Theorem~\ref{thr:models_equal_hintikka}, it
suffices to show that an open tableau for $\theta$ can be turned into
a CMAEHS for $\theta$.

\begin{lemma}
  \label{lm:open_tableau_hintikka}
  If $\tableau{T}^{\theta}$ is open, then there exists a CMAEHS for
  $\theta$.
\end{lemma}

\begin{proofsketch}
  The CMAEHS \hintikka{H} for $\theta$ is built out of the so-called
  \emph{final tree components}.  Each component is a tree-like CMAES
  with nodes labeled with states from $S^{\theta}$, and is associated
  with a state $\D \in S^{\theta}$ and an eventuality $\xi \in
  \cl{\theta}$ (such a component is denoted by $T_{\D,\xi}$).  If $\xi
  \notin \D$, then $T_{\D,\xi}$ is a simple tree, whose root is
  labeled with $\D$, that has exactly one leaf associated with each
  formula $\neg \distrib{A} \psi$ marking an arrow from $\D$ to some
  $\D' \in S^{\theta}$; this leaf is labeled by $\D'$ and connected to
  the root by relation $\rel{R}^D_A$.  If $\xi \in \D$, take the chain
  realizing $\chi$ at $\D$ and give each node ``enough'' successors,
  as prescribed above for simple trees.  The crucial fact is that if
  $\xi'$ is an eventuality in $\D$ that is not ``realized'' inside
  $T_{\D,\xi}$, then $\xi'$ belongs to every leaf of $T_{\D,\xi}$.
  This allows us to stitch up all the $T_{\D,\xi}$'s into a Hintikka
  structure.  The procedure is recursive.  All the eventualities are
  queued. We start from the component uniquely associated with
  $\theta$ (say, we take $T_{\D,\theta}$ where $\D$ is the least
  numbered state containing $\theta$; such a state exists as the
  tableau is open) and then replace each leaf of the structure built
  so far with the component associated with the set marking the leaf
  and the pending eventuality.  The procedure is repeated in cycles
  until we have attached enough components to realize all
  eventualities.  To obtain a CMAEHS, we put $H(\D) = \D$ for all
  $\D$'s.
\end{proofsketch}

\begin{theorem}[Completeness]
  \label{thr:completeness}
  Let $\theta \in \lang$ and let $\tableau{T}^{\theta}$ be open.
  Then, $\theta$ is satisfiable in a CMAEM.
\end{theorem}
\begin{proof}
  Immediate from Lemma~\ref{lm:open_tableau_hintikka} and
  Theorem~\ref{thr:models_equal_hintikka}. 
\end{proof}

As for complexity of the procedure, for lack of space, we only state
that our procedure runs within $\bigo{k^{2n^2}}$ steps, where $n$ is
the size of the input formula and $k$ is the number of agents in the
language. Therefore, the \cmaelcd-satisfiability is in
\cclass{ExpTime}, which together with the \cclass{ExpTime}-hardness
result from~\cite{HM92} for a fragment of our logic containing, along
with individual knowledge modalities, the common knowledge operatore
for the whole set of agents, implies that \cmaelcd-satisfiability is
\cclass{ExpTime}-complete.

\section{Concluding remarks}
\label{sec:concluding}

We have developed a sound and complete, incremental-tableau-based
decision procedure for the full coalitional multiagent epistemic logic
\cmaelcd. We are convinced that this style of tableaux is more
intuitive, practically more efficient and more adaptable than the
top-down style of tableaux e.g., developed for a fragment of this
logic in~\cite{HM92}, and therefore is suitable both for manual and
automated execution. In particular, it is amenable to extension with
operators for strategic abilities of the Alternating-time temporal
logic $\ATL$, a tableaux for which were developed in
\cite{GorSh08}. Merging these two systems is a topic of our future
work.
\section*{Acknowledgments}

We gratefully acknowledge the financial support from the National
Research Foundation of South Africa through a research grant for the
first author, and from the Claude Harris Leon Foundation, funding the
second author's post-doctoral fellowship at the University of the
Witwatersrand, during which this research was done.

%
%
\bibliographystyle{plain}

\appendix

\section{Examples}
\label{sec:examples}

\begin{example}
  Let $\theta = \neg \distrib{\set{a,c}} \common{\set{a,b}} p \land
  \common{\set{a,b}}(p \con q)$, where $\agents = \set{a, b, c}$.  To
  save space, we replace $\theta$ by the set of its conjuncts $\Theta
  = \set{\neg \distrib{\set{a,c}} \common{\set{a,b}} p,
    \common{\set{a,b}}(p \con q)}$. We also use some heuristics that
  we did not have space to discuss in the main part of the paper (they
  will, however, be explained in a follow up work). The picture on the
  left below represents the final pretableau for $\Theta$, while the
  picture on the right represents the initial tableau. Under the
  pictures we list formulae that occur in the labels of states and
  prestates.

  
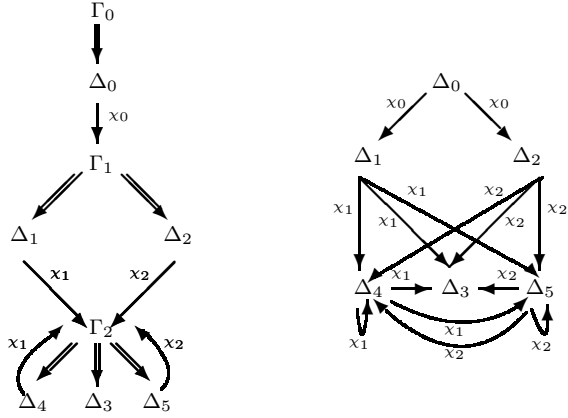
\begin{picture}(200,170)(0,335)
   \footnotesize
    \thicklines

    \put(88,490){\makebox(0,0){
        $\G_0$
      }}

    \put(85,485){\line(0,-1){10}}
    \put(86.25,485){\line(0,-1){10}}
    \put(85.75,475){\vector(0,-1){5}}

    \put(88,463){\makebox(0,0){
        $\D_0$
      }}

    \put(85.75,455){\line(0,-1){10}}
    \put(85.75,445){\vector(0,-1){5}}

    \put(95,450){\makebox(0,0){
        {\tiny $\chi_0$ }
      }}

    \put(89,432){\makebox(0,0){
        {$\G_1$ }
      }}

    \put(77,428){\line(-1,-1){13}}
    \put(80,429){\line(-1,-1){15}}
    \put(66,416){\vector(-1,-1){5}}

    \put(95,428){\line(1,-1){13}}
    \put(97,428){\line(1,-1){13}}
    \put(107.5,416){\vector(1,-1){5}}

    \put(60,405){\makebox(0,0){
        {$\D_1$ }
      }}

    \put(58,395){\line(1,-1){22}}
    \put(77,375){\vector(1,-1){5}}
    \put(73,390){\makebox(0,0){
        {\tiny $\chi_1$ }
      }}

    \put(89,370){\makebox(0,0){
        {$\G_2$ }
      }}

    \put(76,364){\line(-1,-1){10}}
    \put(78,364){\line(-1,-1){10}}
    \put(67.75,355){\vector(-1,-1){5}}

    \put(85,364){\line(0,-1){12}}
    \put(86.5,364){\line(0,-1){12}}
    \put(85.75,353){\vector(0,-1){5}}

    \put(63,342){\makebox(0,0){
        {$\D_4$ }
      }}
    \qbezier(58, 345)(50, 355)(71, 368.9)
    \put(68,366.5){\vector(1,1){5}}
   \put(57,364){\makebox(0,0){
       {\tiny $\chi_1$ }
      }}

    \put(88,342){\makebox(0,0){
        {$\D_3$ }
      }}

    \put(118,405){\makebox(0,0){
        {$\D_2$ }
      }}

    \put(115,395){\line(-1,-1){22}}
    \put(95.7,375){\vector(-1,-1){5}}

    \put(103,390){\makebox(0,0){
        {\tiny $\chi_2$ }
      }}

    \put(91,364){\line(1,-1){10}}
    \put(93,364){\line(1,-1){10}}
    \put(100.75,355){\vector(1,-1){5}}
    \put(110,342){\makebox(0,0){
        {$\D_{5}$ }
      }}
    \put(116,364){\makebox(0,0){
        {\tiny $\chi_2$ }
      }}
    \qbezier(110, 347)(118, 354)(103, 369)
    \put(105,366){\vector(-1,1){5}}


    \put(218,463){\makebox(0,0){
        $\D_0$
      }}

    \put(85.75,455){\line(0,-1){10}}
    \put(85.75,445){\vector(0,-1){5}}

    \put(200,455){\makebox(0,0){
        {\tiny $\chi_0$ }
      }}

    \put(239,455){\makebox(0,0){
        {\tiny $\chi_0$ }
      }}

    \put(210,459){\line(-1,-1){15}}
    \put(197,446){\vector(-1,-1){5}}

    \put(225,459){\line(1,-1){13.2}}
    \put(237.5,446){\vector(1,-1){5}}

    \put(190,435){\makebox(0,0){
        {$\D_1$ }
      }}

    \put(185,427){\vector(0,-1){35}}
    \put(180,415){\makebox(0,0){
        {\tiny $\chi_1$ }
      }}

    \put(185,427){\vector(1,-1){34}}
    \put(197,410){\makebox(0,0){
        {\tiny $\chi_1$ }
      }}

    \qbezier(185, 427)(218, 409)(251,391)
    \put(248,394){\vector(1,-1){5}}
    \put(208,422){\makebox(0,0){
        {\tiny $\chi_1$ }
      }}

    \put(58,395){\line(1,-1){22}}
    \put(77,375){\vector(1,-1){5}}
    \put(73,390){\makebox(0,0){
        {\tiny $\chi_1$ }
      }}

    \put(250,435){\makebox(0,0){
        {$\D_2$ }
      }}

    \put(253,427){\vector(0,-1){35}}
    \put(261,415){\makebox(0,0){
        {\tiny $\chi_2$ }
      }}

    \put(253,427){\vector(-1,-1){34}}
    \put(245,410){\makebox(0,0){
        {\tiny $\chi_2$ }
      }}

    \qbezier(254, 427)(222, 408)(190,389)
    \put(193,392){\vector(-1,-1){5}}
    \put(237,422){\makebox(0,0){
        {\tiny $\chi_2$ }
      }}

    \put(58,395){\line(1,-1){22}}
    \put(77,375){\vector(1,-1){5}}
    \put(73,390){\makebox(0,0){
        {\tiny $\chi_1$ }
      }}


    \put(115,395){\line(-1,-1){22}}
    \put(95.7,375){\vector(-1,-1){5}}

    \put(103,390){\makebox(0,0){
        {\tiny $\chi_2$ }
      }}

    \put(76,364){\line(-1,-1){10}}
    \put(78,364){\line(-1,-1){10}}
    \put(67.75,355){\vector(-1,-1){5}}

    \put(85,364){\line(0,-1){12}}
    \put(86.5,364){\line(0,-1){12}}
    \put(85.75,353){\vector(0,-1){5}}

    \put(190,385){\makebox(0,0){
        {$\D_4$ }
      }}

    \put(197,385){\vector(1,0){15}}
   \put(202,390){\makebox(0,0){
       {\tiny $\chi_1$ }
     }}

   \qbezier(188, 378)(186, 358)(184, 377)
   \put(188,376){\vector(0,1){5}}
   \put(186,364){\makebox(0,0){
       {\tiny $\chi_1$ }
     }}

    \qbezier(197, 380)(225, 365)(246, 379)
    \put(244,377){\vector(1,1){5}}
   \put(222,360){\makebox(0,0){
       {\tiny $\chi_2$ }
     }}

    \qbezier(195, 376.5)(225, 350)(248, 376)
    \put(195,376){\vector(-1,1){5}}
   \put(222,369){\makebox(0,0){
       {\tiny $\chi_1$ }
     }}

   \put(57,364){\makebox(0,0){
       {\tiny $\chi_1$ }
      }}

    \put(223,385){\makebox(0,0){
        {$\D_3$ }
      }}

    \put(91,364){\line(1,-1){10}}
    \put(93,364){\line(1,-1){10}}
    \put(100.75,355){\vector(1,-1){5}}
    \put(255,385){\makebox(0,0){
        {$\D_{5}$ }
      }}

   \qbezier(256, 378)(257, 358)(250, 377)
   \put(256,374){\vector(0,1){5}}
   \put(255,363){\makebox(0,0){
       {\tiny $\chi_2$ }
     }}

   \put(245,385){\vector(-1,0){15}}
   \put(242,390){\makebox(0,0){
       {\tiny $\chi_2$ }
      }}
    \put(116,364){\makebox(0,0){
        {\tiny $\chi_2$ }
      }}
    \qbezier(110, 347)(118, 354)(103, 369)
    \put(105,366){\vector(-1,1){5}}
  \end{picture}

  {\small

    \noindent $\chi_0 = \neg \distrib{\set{a,c}} \common{\set{a,b}} p$;
    \noindent $\chi_1 = \neg \distrib{a} (p \con \common{\set{a,b}} p)$;
    \noindent $\chi_2 = \neg \distrib{b} (p \con \common{\set{a,b}} p)$;

    \noindent $\G_0 = \set{\neg \distrib{\set{a,c}} \common{\set{a,b}} p,
      \common{\set{a, b}} (p \con q)}$;

    \noindent $\D_0 = \set{\neg \distrib{\set{a,c}} \common{\set{a,b}}
      p, \distrib{a} [(p \con q) \con \common{\set{a, b}} (p \con q)],
      \distrib{b} [(p \con q) \con \common{\set{a, b}} (p \con q)]}$;

    \noindent $\G_1 = \set{\neg \common{\set{a,b}} p, \distrib{a} [(p
      \con q) \con \common{\set{a, b}} (p \con q)]}$;

    \noindent $\D_1 = \set{\neg \distrib{a} (p \con \common{\set{a,
          b}} p), \distrib{a} [(p \con q) \con \common{\set{a, b}} (p
      \con q)]}$;

    \noindent $\D_2 = \set{\neg \distrib{b} (p \con \common{\set{a,
          b}} p), \distrib{b} [(p \con q) \con \common{\set{a, b}} (p
      \con q)]}$;

    \noindent $\G_2 = \set{\neg(p \con \common{\set{a, b}} p), (p \con
      q) \con \common{\set{a, b}} (p \con q)}$;

    \noindent $\D_3 = \set{\neg p, (p \con q) \con \common{\set{a, b}}
      (p \con q), p \con q, p, q, \distrib{a} [(p \con q) \con
      \common{\set{a, b}} (p \con q)],$

      $\distrib{b} [(p \con q) \con \common{\set{a, b}} (p \con q)]}$;

    \noindent $\D_4 = \set{\neg \common{\set{a, b}} p, \neg \distrib{a}
      (p \con \common{\set{a, b}} p), (p \con q) \con \common{\set{a,
          b}} (p \con q), p \con q, p, q,$

      $\distrib{a} [(p \con q) \con \common{\set{a, b}} (p \con q)],
      \distrib{b} [(p \con q) \con \common{\set{a, b}} (p \con q)]}$;

    \noindent $\D_5 = \set{\neg \common{\set{a, b}} p, \neg \distrib{b}
      (p \con \common{\set{a, b}} p), (p \con q) \con \common{\set{a,
          b}} (p \con q), p \con q, p, q,$

      $\distrib{a} [(p \con q) \con \common{\set{a, b}} (p \con q)],
      \distrib{b} [(p \con q) \con \common{\set{a, b}} (p \con q)]}$.

    }

    \medskip During the state-elimination phase, the state $\D_3$ is
    removed due to \Rule{E1}, as it contains a patent inconsistency
    ($p, \neg p$). Then, the states $\D_1, \D_2, \D_4$, and $\D_5$ are
    eliminated due to \Rule{E3}, as all of them contain the unrealized
    eventuality $\neg \common{\set{a,b}} p$.  Finally, $\D_0$ gets
    eliminated, as it has lost all its successors along the arrow
    marked with $\chi_0$.  Thus, the final tableau for $\Theta$ is an
    empty graph; therefore, $\Theta$ is \emph{unsatisfiable}.
\end{example}

\begin{example}
  Let $\theta = \common{\set{a, b}} p \con \common{\set{b, c}} p \con
  \neg \common{\set{a, c}} p$, where $\agents = \set{a, b, c}$. Once
  again, to save space, we replace $\theta$ with $\Theta =
  \set{\common{\set{a, b}} p, \common{\set{b, c}} p, \neg
    \common{\set{a, c}} p}$.  The picture below shows the final
  pretableau for $\Theta$; the initial tableau is easily extracted
  from it, as in the previous example. Formulae that occur in the
  labels of states and prestates are listed under the picture. 

\medskip
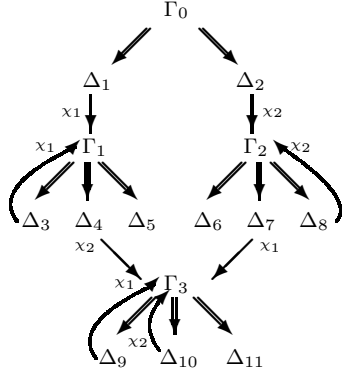
\begin{picture}(220,145)(-40,290)
    \footnotesize
    \thicklines
    \put(100,432){\makebox(0,0){
        {$\G_0$ }
      }}

    \put(87,425){\line(-1,-1){10}}
    \put(88.5,425){\line(-1,-1){10}}
    \put(78.75,416){\vector(-1,-1){5}}

    \put(107,425){\line(1,-1){10}}
    \put(108.5,425){\line(1,-1){10}}
    \put(116.75,416){\vector(1,-1){5}}

    \put(70,405){\makebox(0,0){
        {$\D_1$ }
      }}

    \put(66,400){\line(0,-1){15}}
    \put(66,391){\vector(0,-1){5}}


    \put(60,393){\makebox(0,0){
        {\tiny $\chi_1$ }
      }}

    \put(128,405){\makebox(0,0){
        {$\D_2$ }
      }}

    \put(127,400){\line(0,-1){15}}
    \put(127,391){\vector(0,-1){5}}

    \put(136,393){\makebox(0,0){
        {\tiny $\chi_2$ }
      }}

    \put(69,380){\makebox(0,0){
        {$\G_1$ }
      }}

    \put(130,380){\makebox(0,0){
        {$\G_2$ }
      }}

    \put(58,374){\line(-1,-1){10}}
    \put(59.5,374){\line(-1,-1){10}}
    \put(49.75,365){\vector(-1,-1){5}}

    \put(47,352){\makebox(0,0){
        {$\D_3$ }
      }}

    \qbezier(38, 352)(28, 365)(60, 380)
    \put(57,378){\vector(1,1){5}}
    \put(50,380){\makebox(0,0){
        {\tiny $\chi_1$ }
      }}

    \put(64,374){\line(0,-1){10}}
    \put(65.25,374){\line(0,-1){10}}
    \put(64.75,365){\vector(0,-1){5}}

    \put(67,352){\makebox(0,0){
        {$\D_4$ }
      }}

    \cut{\qbezier(60, 355)(54, 365)(63, 375)
    \put(61,373){\vector(1,1){5}}
    \put(55,358){\makebox(0,0){
        {\tiny $\chi_1$ }
      }}}

    \put(70, 345){\vector(1,-1){15}}
    \put(65,342){\makebox(0,0){
        {\tiny $\chi_2$ }
      }}

    \put(69,374){\line(1,-1){10}}
    \put(70.5,374){\line(1,-1){10}}
    \put(78.75,365){\vector(1,-1){5}}

    \put(87,352){\makebox(0,0){
        {$\D_5$ }
      }}

    \put(123,374){\line(-1,-1){10}}
    \put(124.5,374){\line(-1,-1){10}}
    \put(114.75,365){\vector(-1,-1){5}}

    \put(112,352){\makebox(0,0){
        {$\D_6$ }
      }}

    \put(129,374){\line(0,-1){10}}
    \put(130.25,374){\line(0,-1){10}}
    \put(129.75,365){\vector(0,-1){5}}

    \put(132,352){\makebox(0,0){
        {$\D_7$ }
      }}

    \cut{\qbezier(133, 355)(139, 365)(132, 375)
    \put(134,372){\vector(-1,1){5}}
    \put(141.5,358){\makebox(0,0){
        {\tiny $\chi_2$ }
      }}}

    \put(134,374){\line(1,-1){10}}
    \put(135.5,374){\line(1,-1){10}}
    \put(143.75,365){\vector(1,-1){5}}

    \put(152,352){\makebox(0,0){
        {$\D_8$ }
      }}

    \put(146.5,380){\makebox(0,0){
        {\tiny $\chi_2$ }
      }}
    \qbezier(159, 352)(166, 361)(138, 380)
    \put(140,378){\vector(-1,1){5}}

    \put(127,345){\vector(-1,-1){15}}
    \put(135,342){\makebox(0,0){
        {\tiny $\chi_1$ }
      }}

    \put(100,328){\makebox(0,0){
        {$\G_3$ }
      }}

    \put(89,323){\line(-1,-1){10}}
    \put(90.5,323){\line(-1,-1){10}}
    \put(80.75,314){\vector(-1,-1){5}}

    \put(76,300){\makebox(0,0){
        {$\D_9$ }
      }}
    \qbezier(68, 300)(58, 313)(90, 328)
    \put(87,326){\vector(1,1){5}}
    \put(80,328){\makebox(0,0){
        {\tiny $\chi_1$ }
      }}

    \put(97,323){\line(0,-1){10}}
    \put(98.5,323){\line(0,-1){10}}
    \put(97.75,314){\vector(0,-1){5}}

    \put(101,300){\makebox(0,0){
        {$\D_{10}$ }
      }}
    \qbezier(92, 303)(84, 313)(93, 323)
    \put(91,321){\vector(1,1){5}}
    \put(85,306){\makebox(0,0){
        {\tiny $\chi_2$ }
      }}

    \put(105,323){\line(1,-1){10}}
    \put(106.5,323){\line(1,-1){10}}
    \put(114.75,314){\vector(1,-1){5}}

    \put(126,300){\makebox(0,0){
        {$\D_{11}$ }
      }}

    
    \cut{\qbezier(97, 305)(95, 313)(80, 325)
    \put(80,324){\vector(-1,1){5}}
    \put(88,325){\makebox(0,0){
        {\tiny $\chi_2$ }
      }}}

    \cut{ \qbezier(97, 305)(100, 313)(115, 325)
    \put(113,322){\vector(1,1){5}}
    \put(108 ,324){\makebox(0,0){
        {\tiny $\chi_1$ }
      }}

    \put(127,328){\makebox(0,0){
        {$\G_4$ }
      }}

    \put(119,323){\line(-1,-1){10}}
    \put(120.5,323){\line(-1,-1){10}}
    \put(110.75,314){\vector(-1,-1){5}}

    \put(124,323){\line(0,-1){10}}
    \put(125.5,323){\line(0,-1){10}}
    \put(124.75,314){\vector(0,-1){5}}

    \put(125,300){\makebox(0,0){
        {$\D_{12}$ }
      }}
    \qbezier(130, 303)(136, 313)(129, 323)
    \put(131,320){\vector(-1,1){5}}
    \put(138.5,306){\makebox(0,0){
        {\tiny $\chi_1$ }
      }}

    \put(130,323){\line(1,-1){10}}
    \put(131.5,323){\line(1,-1){10}}
    \put(139.75,314){\vector(1,-1){5}}

    \put(150,300){\makebox(0,0){
        {$\D_{13}$ }
      }}

    \put(146,328){\makebox(0,0){
        {\tiny $\chi_2$ }
      }}
    \qbezier(159, 300)(166, 309)(135, 327)
    \put(137.5,325){\vector(-1,1){5}}}
  \end{picture}

  {\small

    \noindent $\chi_1 = \neg \distrib{a} (p \con \common{\set{a, c}}
    p)$; $\chi_2 = \neg \distrib{c} (p \con \common{\set{a, c}} p)$;

  \noindent $\G_0 = \set{\common{\set{a, b}} p, \common{\set{b, c}} p,
    \neg \common{\set{a, c}} p}$;

  \noindent $\D_1 = \set{\common{\set{a, b}} p, \common{\set{b, c}} p,
    \neg \common{\set{a, c}} p, \distrib{a} (p \con \common{\set{a,
        b}}p), \distrib{b} (p \con \common{\set{a, b}}p), \distrib{b}
    (p \con \common{\set{b, c}}p),$

    $\distrib{c} (p \con \common{\set{b, c}}p), \neg \distrib{a} (p
    \con \common{\set{a, c}}p)}$;

  \noindent $\D_2 = \set{\common{\set{a, b}} p, \common{\set{b, c}} p,
    \neg \common{\set{a, c}} p, \distrib{a} (p \con \common{\set{a,
        b}}p), \distrib{b} (p \con \common{\set{a, b}}p), \distrib{b}
    (p \con \common{\set{b, c}}p),$

    $\distrib{c} (p \con \common{\set{b, c}}p), \neg \distrib{c} (p
    \con \common{\set{a, c}}p)}$;

  \noindent $\G_1 = \set{\neg (p \con \common{\set{a, c}} p),
    \distrib{a} (p \con \common{\set{a, b}}p)}$;

  \noindent $\G_2 = \set{\neg (p \con \common{\set{a, c}} p),
    \distrib{c} (p \con \common{\set{b, c}}p)}$;

  \noindent $\D_3 = \set{ \neg \common{\set{a, c}} p, \distrib{a} (p
    \con \common{\set{a, b}}p), \neg \distrib{a} (p \con
    \common{\set{a, c}}p)}$;

  \noindent $\D_4 = \set{\neg \common{\set{a, c}} p, \distrib{a} (p
    \con \common{\set{a, b}} p), \neg \distrib{c} (p \con
    \common{\set{a, c}} p}$;

  \noindent $\D_5= \set{\neg p, \distrib{a} (p \con \common{\set{a,
        b}}p), p \con \common{\set{a, b}}p, p, \common{\set{a, b}}p}$;

  \noindent $\D_6= \set{\neg p, \distrib{c} (p \con \common{\set{b,
        c}}p), p \con \common{\set{b, c}}p, p, \common{\set{b, c}}p}$;

  \noindent $\D_7 = \set{\neg \common{\set{a, c}} p, \distrib{c} (p
    \con \common{\set{b, c}} p), \neg \distrib{a} (p \con
    \common{\set{a, c}} p}$;

  \noindent $\D_8 = \set{\neg \common{\set{a, c}} p, \distrib{c} (p
    \con \common{\set{b, c}} p), \neg \distrib{c} (p \con
    \common{\set{a, c}} p}$;

  \noindent $\G_3 = \set{\neg(p \con \common{\set{a, c}} p)}$

  \noindent $\D_{9} = \set{\neg \common{\set{a, c}} p, \neg
    \distrib{a} (p \con \common{\set{a, c}} p)}$;

  \noindent $\D_{10} = \set{\neg \common{\set{a, c}} p, \neg
    \distrib{c} (p \con \common{\set{a, c}})}$;

  \noindent $\D_{11} = \set{\neg p}$.}

\medskip

At the state elimination phase, states $\D_5$ and $\D_6$ get removed
due to \Rule{E1}.  All other states remain in place; in particular, no
states gets eliminated due to \Rule{E3}, because, from any state one
can reach either $\D_9$ or $\D_{13}$, both of which contain $\neg p$,
and our only eventuality is $\neg \common{\set{a, c}} p$. Thus,
$\theta$ is \emph{satisfiable}, and a Hintikka structure for it is
readily extracted from the final tableau.
\end{example}

\end{document}